\documentclass[sn-mathphys-num]{sn-jnl}
\usepackage{amsmath}
  \usepackage{graphics} 
  \usepackage{epsfig} 
\usepackage{graphicx}  \usepackage{epstopdf}

\usepackage{latexsym}
\usepackage{amsmath,amssymb}
\usepackage{amsmath}
\usepackage{xcolor}
\usepackage{enumerate}

\usepackage{float}
\newenvironment{proof}{\textbf{Proof:}}{\hfill$\square$}

\newtheorem{theo}{Theorem}[section]
\newtheorem{coll}[theo]{Corollary}
\newtheorem{lem}[theo]{Lemma}
\newtheorem{prop}[theo]{Proposition}
\newtheorem{defn}[theo]{Definition}
\newtheorem{ex}[theo]{Example}
\newtheorem{rem}[theo]{Remark}

\begin{document}

\title[Double skew cyclic codes over $\mathbb{F}_q+v\mathbb{F}_q$]{Double skew cyclic codes over $\mathbb{F}_q+v\mathbb{F}_q$}

\author*[1]{\fnm{Ashutosh} \sur{Singh}}\email{ashutosh\_1921ma05@iitp.ac.in}

\author[2]{\fnm{Tulay} \sur{Yildirim}}\email{tulayturan@karabuk.edu.tr}

\author[1]{\fnm{Om} \sur{Prakash}}\email{om@iitp.ac.in}

\affil*[1]{\orgdiv{Department of Mathematics}, \orgname{Indian Institute of Technology Patna}, \orgaddress{\city{Patna}, \postcode{801 106}, \country{India}}}

\affil[2]{\orgdiv{Eskipazar Vocational School}, \orgname{Karab$\ddot{u}$k University}, \orgaddress{\street{Street}, \city{Karab$\ddot{u}$k}, \country{Turkey}}}


\abstract{In this study, in order to get better codes, we focus on double skew cyclic codes over the ring $\mathrm{R}= \mathbb{F}_q+v\mathbb{F}_q, ~v^2=v$ where $q$ is a prime power. We investigate the generator polynomials, minimal spanning sets, generator matrices, and the dual codes over the ring $\mathrm{R}$. As an implementation, the obtained results are illustrated with some good examples. Moreover, we introduce a construction for new generator matrices and thus achieve codes with improved parameters compared to those found in existing literature. Finally, we tabulate our obtained block codes over the ring $\mathrm{R}$.}




\maketitle

\section{Introduction}

Codes over finite rings were introduced in the early seventies. The class of cyclic codes over rings is one of the most important classes of linear codes that have attracted much attention due to their rich algebraic properties and successful applications in combinatorial coding. Many authors have studied these codes over fields, finite chain rings and other algebraic structures for different contexts \cite{Alahmadi,A10, Dougherty, A11,  A12,  A13, Patel}. In 2007, Boucher et al. \cite{A9} initiated the study of cyclic codes using a noncommutative ring $\mathbb{F}[x;\theta]$ where $\theta$ is an automorphism of the finite field $\mathbb{F}$, and the authors have produced many numerical examples which improved the best-known codes. Since the skew polynomial ring is a non-UFD (unique factorization domain), the polynomial $x^n-1$ has the advantage of possessing additional factors in the respective skew polynomial ring.

On the other hand, the structure of double cyclic codes over finite rings has been studied by many mathematicians and scientists \cite{A7,A6,A17,A15}. The study of double cyclic codes was initiated by Ayats et al. \cite{Ayats} in 2014. They defined $\mathbb{Z}_2$-double cyclic codes as a subfamily of binary linear codes and investigated the algebraic structure of $\mathbb{Z}_2$-double cyclic codes and their dual codes. In 2016, Gao et al. \cite{A16} characterized the double cyclic codes over $\mathbb{Z}_4$, and they obtained some optimal codes. Later, Bathala et al. \cite{A14} determined the generating polynomials of $\mathbb{F}_4[v]$-double cyclic codes and their duals. Also, they investigated some algebraic properties of double constacyclic codes over $\mathbb{F}_4[v]$. In 2020, Deng et al. \cite{A6} studied double cyclic codes over the ring $\mathbb{F}_q+v\mathbb{F}_q$ and obtained some better codes. Recently, Aydogdu et al. \cite{A7} characterized the algebraic structure of double skew cyclic codes over the finite field $\mathbb{F}_q$. They also illustrated that MDS codes can be obtained with double skew cyclic codes over $\mathbb{F}_q$.
The above study motivates us to delve into the double skew cyclic codes over the ring $\mathrm{R}:=\mathbb{F}_q+v\mathbb{F}_q$, $v^2=v$, which is a finite non-chain ring. Here, we use the most fundamental ideas and theories by overcoming some computation difficulties to get better codes over the considered ring. For example, unlike the canonical maps used in the literature, we defined the Gray map over $\mathrm{R}$ by using invertible matrices over $\mathbb{F}_q$. Because we observe that the distance of a code over the ring $\mathrm{R}$ gives better results using matrix Gray maps instead of canonical Gray maps. We also generate new codes by constructing new generator matrices over the ring $\mathrm{R}$, which allowed us to obtain new codes with better parameters than existing parameters in the literature.

This paper is organized as follows: Section 2 contains some basic definitions and theories, and we also recall some structural properties of skew cyclic codes over $\mathrm{R}$. In Section 3, we present the generator polynomials of $\mathrm{R}$-double skew cyclic codes and studied minimal generating sets and generating matrices over $\mathrm{R}$. Section 4 describes the generating polynomial and the relationship between the codes and their dual. Finally, Section 5 provides some essential examples of optimal parameters double skew cyclic codes of block length over the ring $\mathrm{R}$. Here, we emphasize this work by introducing a new generator matrix to obtain new codes with better parameters.
\section{Preliminaries and Definitions}
Let $\mathbb{F}_q$ be the finite field with $q$ elements and $\mathrm{R}:=\mathbb{F}_q+v\mathbb{F}_q=\lbrace   a+vb \vert  a,  b\in \mathbb{F}_q   \rbrace$  with $v^2=v$. Thus, the ring $\mathrm{R}$ is a non-chain and semi-local ring with $\langle v\rangle$  and $\langle 1-v\rangle$ as maximal ideals. By the Chinese Remainder Theorem, an element of $\mathrm{R}$ can uniquely be expressed as $a+vb=(a+b)v+a(1-v)$ for all $a,  b\in \mathbb{F}_q$.\\

Recall that a code $C$ of length $n$ over a ring $\mathrm{R}$ is a non-empty subset of $\mathrm{R}^n$. A  linear code $C$ of length $n$ over $\mathrm{R}$ is an $\mathrm{R}$-submodule of $\mathrm{R}^n$.  The dual code of $C$,  denoted by $C^\bot$,  is also an $\mathrm{R}$-linear code and defined as $C^{\bot}= \lbrace y\in \mathrm{R}^n : <x,  y>=0,   \forall x\in C \rbrace$ where $<x,  y>$ is  an inner product of $x$ and $y$ in $\mathrm{R}^n$.\\

The Frobenius automorphism of $\mathbb{F}_q$ is the function $\zeta:\mathbb{F}_q \rightarrow \mathbb{F}_q$ defined by $\zeta(a)=a^p$ where $p$ is a prime. We consider the automorphism $\theta$ of $\mathrm{R}$, given in \cite{A2}, defined by
\begin{align*}
a+vb &\mapsto a^{p^i}+vb^{p^i}.
\end{align*}
Obviously, for $i$=1, $\theta_i$ is the Frobenius automorphism of $\mathbb{F}_q$. The set $\mathrm{R}[x; \theta_i] =\lbrace  \sum_{i=0}^nr_ix^i \vert r_i\in \mathrm{R} \rbrace$ of polynomials over $\mathbb{F}_q$ forms a ring under the usual addition of polynomials, and multiplication defined with respect to the rule $xr=\theta_i(r)x$ for all $r\in \mathrm{R}$. Clearly,  $\mathrm{R}[x; \theta_i]$ is a noncommutative ring unless $\theta_i$ is the identity automorphism. This ring is known as the skew polynomial ring over $\mathrm{R}$, and its elements are called skew polynomials.
\begin{defn}
Let $\vartheta$ be a map given by
\begin{align*}
\vartheta\colon \mathrm{R}[x; \theta_i] &\to \mathrm{R}[x; \theta_i] \\
\sum_{i=0}^nr_ix^i &\mapsto \sum_{i=0}^n\theta_i(r_i)x^i,
\end{align*}
where $r_i\in \mathrm{R}$. Then, $\vartheta$ is a ring homomorphism.
\end{defn}
\begin{defn}
Let $\theta_i$ be  the automorphism  of $\mathrm{R}$ and $ax^m,   bx^n \in \mathrm{R}[x;  \theta_i]$. Then the multiplication rule of $\mathrm{R}[x;\theta_i]$ skew polynomial ring over $\mathrm{R}$ is
\begin{align*}
(ax^m)\cdot(bx^n)=a\theta_i^m(b)x^{m+n}.
\end{align*}
\end{defn}

A subset $C$ of $\mathrm{R}^n$ is said to be an $\mathrm{R}$-linear skew cyclic code of length $n$ if $C$ is an $\mathrm{R}$-submodule of $\mathrm{R}^n$, and $C$ is closed under the ${\theta_i}$-cyclic shift, i.e.  for $c=(c_0, c_1, \ldots , c_{n-1})\in C$,
${\theta_i}(c)=(\theta_i(c_{n-1}), \theta_i(c_0), \ldots , \theta_i(c_{n-2})) \in C$. Note that each codeword $c$ of $C$ can be defined with its polynomial representation $c(x)=c_0+c_1x+ \ldots + c_{n-1}x^{n-1}$.

For convenience, we set $v^\prime:=1-v$ and also define a polynomial $f(x)\in \mathrm{R}[x;\theta_i]$ as $f(x):=f_v(x)v+f_{v^\prime}(x)v^\prime$ where $f_v(x),  f_{v^\prime}(x) \in \mathbb{F}_q[x;\theta_i]$. \\

 \begin{lem}\label{D1}  \cite{A4}
 Let $f(x),  g(x)\in \mathrm{R}[x; \theta_i]$ with $f(x)\neq 0$. Then
  there exists $q(x),  r(x )\in \mathrm{R}[x; \theta_i]$ such that $g(x)=q(x)f(x)+r(x)$  where $r(x)=0$ or $deg(r(x))< deg(f(x))$.  Moreover, $ \mathrm{R}[x; \theta_i]$ is  a principal ideal domain.
 \end{lem}
 The below lemmas can be proved with slight modifications to the associated theories in \cite{A6}.
 \begin{lem}
Let $f(x)$ be a function in $\mathrm{R}[x;\theta_i]$. Then, $f(x)$ is called right divisor of $g(x)$ if there exists a skew polynomial $h(x)$ such that $g_v(x)=h_v(x)f_v(x)$ and $g_{v^\prime}(x)=h_{v^\prime}(x)f_{v^\prime}(x)$, denoted by $f(x)\vert_r g(x)$.
  \end{lem}

\begin{lem}
  Let $f(x)$ and $g(x)$ be skew polynomials in $\mathrm{R}[x;\theta_i]$. The greatest common divisor of $f(x)$ and $g(x)$ is the monic polynomial $d(x)\in \mathrm{R}[x;\theta_i]$ such that $d(x)=gcd(f_v(x), g_v(x))+gcd(f_{v^\prime}(x), g_{v^\prime}(x))$.

\end{lem}

\begin{lem}
   The least common multiple of skew polynomials $f(x)$ and $g(x)$ in $\mathrm{R}[x;\theta_i]$ is the unique monic polynomial $\gamma(x)$ such that $f(x)\vert_r \gamma(x)$ and $g(x)\vert_r \gamma(x)$ and for any $\gamma^\prime(x) \in \mathrm{R}[x;\theta_i]$ with $f(x)\vert_r \gamma^\prime(x)$ and $g(x)\vert_r \gamma^\prime(x)$, we have $\gamma(x)\vert_r \gamma^\prime(x)$. We denote $lcm(f(x),  g(x))$  by $\gamma(x)$.
\end{lem}


\begin{defn}  \cite{gao}
     Let $C$ be a linear code of length $n$ over $\mathrm{R}$. The Gray map $ \phi: \mathrm{R} \rightarrow \mathbb{F}_q^2$ is defined by
     $$
\begin{aligned}
& \phi(\mathbf{a}) = (a_0', a_1')N,
\end{aligned}
$$
where $\mathbf{a}=(a_0', a_1')$ with $\mathbf{a}=a_0+a_1 v=v a_0'+(1-v) a_1' \in \mathrm{R}$, for $a_0,  a_1,  a_0^\prime,  a_1^\prime \in \mathbb{F}_q$ and $ N \in G L_{2}\left(\mathbb{F}_q\right)$ with $NN^T=\eta I_2$ for $\eta \in \mathbb{F}_q-\{0\}$.
\end{defn}

 Clearly, $\phi$ is an $\mathbb{F}_q$-module isomorphism. For brevity, we write the vector $(a_0', a_1')N$ as $\mathbf{a} N$. This map naturally can be extended as


\begin{align*}
\phi \colon \mathrm{R}^{n} &\to \mathbb{F}_q^{2n} \\
(\mathbf{a}_0, \mathbf{a}_1, \ldots , \mathbf{a}_{n-1})  &\mapsto (\mathbf{a}_0 N, \mathbf{a}_1 N, \ldots, \mathbf{a}_{n-1} N).
\end{align*}
For any element $\mathbf{a} \in \mathrm{R}$, we denote the Hamming weight $w_H(\mathbf{a} N)$ of $\mathbf{a} N$ as the Gray weight of $\mathbf{a}$, i.e., $w_G(\mathbf{a})=w_H(\mathbf{a} N)$. The Gray weight of any element $(\mathbf{a}_0, \mathbf{a}_1, \ldots , \mathbf{a}_{n-1}) \in \mathrm{R}^n$ is equal to the integer $\sum_{i=0}^{n-1}w_G(\mathbf{a}_i)$.

\begin{rem}
    Throughout the paper, for the underlined fields $N$ is $\left[\begin{array}{cc}
     1&t  \\
     t&1
\end{array} \right]$ for  even characteristics or $\left[\begin{array}{cc}
     1&1  \\
     1&-1
\end{array} \right]$ for   odd characteristics, where $t$ is a primitive root of unity in relevant field.
\end{rem}

\section{Double skew cyclic codes over $\mathrm{R}$}

Let $r$ and $s$ be positive integers such that $n=r+s$, i.e., the $n$-coordinates of each $n$-tuple in $\mathrm{R}^n$ can be partitioned into two sets of sizes $r$ and $s$. Then, $\mathrm{R}^n$ can be defined as $\mathrm{R}$-submodule of $\mathrm{R}^r\times \mathrm{R}^s$.  Any linear code $C$ of length $n$ over $\mathrm{R}$ is an $\mathrm{R}$-submodule of $\mathrm{R}^r\times \mathrm{R}^s$. The set of all double skew cyclic codes of length $(r,  s)$ over $\mathrm{R}$ is denoted by $C_{r,  s}(\mathrm{R})$. Throughout the paper,  $C\in C_{r,  s}(\mathrm{R})$.\\

\begin{defn}
Let $\theta_i$ be an automorphism of $\mathrm{R}$ and $C,$ a skew cyclic code of length $(r,  s)$ over $\mathrm{R}$.  The code $C$ is called a double skew cyclic code of length $(r,  s)$ over $\mathrm{R}$ if $C$ is closed under the $T_{\theta_i}$-double cyclic shift,  i.e., for $c=(c_0,  c_1,  \ldots,  c_{r-1} \vert  c_0^\prime,  c_1^\prime,  \ldots,  c_{s-1}^\prime)\in C$,
\begin{align*}
T_{\theta_i}(c)=(\theta_i(c_{r-1}),  \theta_i(c_0),  \ldots,  \theta_i(c_{r-2})\vert \theta_i(c^{\prime}_{s-1}),  \theta_i(c^{\prime}_0),  \ldots,  \theta_i(c^{\prime}_{s-2}) )\in C.
\end{align*}
\end{defn}

 Each $(c_0,  c_1,  \ldots,  c_{r-1} \vert  c_0^\prime,  c_1^\prime,  \ldots,  c_{s-1}^\prime)\in \mathrm{R}^r\times \mathrm{R}^s$ can be identify with a pair of polynomials $(c(x)|c^\prime(x))=(c_0+ c_1x+ \ldots +  c_{r-1}x^{r-1} \vert  c_0^\prime+ c_1^\prime x \ldots + c_{s-1}^\prime x^{s-1})$ in $\mathcal{R}_{r,  s}=\mathcal{R}_r \times \mathcal{R}_s=\frac{\mathrm{R}[x; \theta_i]}{\langle x^r-1\rangle}\times \frac{\mathrm{R}[x; \theta_i]}{\langle x^s-1\rangle }$.   Clearly,  there is a one-to one correspondence between $\mathrm{R}^r\times \mathrm{R}^s$  and $\mathcal{R}_{r,  s}$.  Moreover,  for any $a\in \mathrm{R}$ and $(\textbf{m}|\textbf{n})\in \mathrm{R}^r\times \mathrm{R}^s$,  we define the product $a\cdot(\textbf{m}|\textbf{n})=(a\textbf{m}|a\textbf{n})$.  \\

The dual of $C$ is an $\mathrm{R}$-double skew cyclic code of length $(r, s)$ defined as $C^\bot=\lbrace x\in \mathrm{R}^r \times \mathrm{R}^s : <x,y> = 0,  \forall y\in C \rbrace$.\\

Let $f(x)\in \mathrm{R}[x; \theta_i]$,  and  $a_0,  a_1,  \ldots,  a_n$  coefficients of the polynomial $f(x)$. Then the polynomial $f^\star(x)=\sum_{j=0}^{s}\theta_i^j(a_{s-j})x^j$ is called the reciprocal polynomial of $f(x)$.\\

 \begin{theo}\label{D2}
A code $C$ is an $\mathrm{R}$-double skew cyclic code of length $(r,  s)$ if and only if $C$ is a left $\mathrm{R}[x; \theta_i]$-submodule of $\mathcal{R}_{r,  s}$.
\end{theo}
\begin{proof}     This can be proved easily by the definition of double-cyclic codes.
\end{proof}


\begin{theo}\label{D3}
Let $C$ be a double skew cyclic code of length $(r,  s)$ over $\mathrm{R}$. Then
\begin{align*}
C=\langle (g(x)\vert 0),  (l(x)\vert h(x))\rangle ,
\end{align*}
where $g(x),  l(x) \in \mathcal{R}_r$, $h(x)\in \mathcal{R}_s$ are monic polynomials with $g(x)=g_v(x)v+g_{v^\prime}(x)v^\prime$,   $l(x)=l_v(x)v+l_{v^\prime}(x)v^\prime $,  and  $h(x)=h_v(x)v+h_{v^\prime}(x)v^\prime$  such that $g_v(x),  g_{v^\prime}(x)|_r(x^r-1)$ and $h_v(x),  h_{v^\prime}(x)|_r(x^s-1)$.
\end{theo}

\begin{proof} Let $C$ be a left $\mathrm{R}[x;\theta_i]$-submodule of $\mathcal{R}_{r,s}$.  Define a projective homomorphism of a left $\mathrm{R}[x;\theta_i]$-modules as
\begin{align*}
\sigma \colon C &\to  \mathcal{R}_s \\
(v_1(x),  v_2(x)) &\mapsto v_2(x).
\end{align*}
From Theorem 12 of \cite{A3}, $Im(\sigma)=\langle h(x)\rangle $  where $h(x)=h_v(x)v+h_{v^\prime}(x)v^\prime\in \mathcal{R}_s$ with $h_v(x)|_r(x^s-1)$ and $h_{v^\prime}(x)|_r(x^s-1)$.  We know that
\begin{align*}
Ker(\sigma)=\lbrace   (v_1(x)\vert 0)\in C : v_1(x)\in \mathcal{R}_r \rbrace .
\end{align*}
Define
\begin{align*}
S=\lbrace  r(x)\in \mathcal{R}_r : (r(x)\vert 0)\in Ker(\sigma)   \rbrace .
\end{align*}
Clearly,  $S$ is a left $\mathrm{R}[x; \theta_i]$-submodule of $\mathcal{R}_r$.  By \cite{A3},  $S=\langle g(x)\rangle $,  where $g(x)=g_v(x)v+g_{v^\prime}(x)v^\prime \in \mathcal{R}_r$.  For any element $(r(x)\vert 0)$ in $Ker(\sigma)$,  we have $r(x)\in S =\langle g(x)\rangle $, where $r(x)=r_v(x)v+r_{v^\prime}(x)v^\prime$.  This means that there exists a polynomial $\Bar{h}(x)\in \mathcal{R}_r$ such that $r_v(x)=\Bar{h}_v(x).g_v(x)$ and $r_{v^\prime}(x)=\Bar{h}_{v^\prime}(x).g_{v^\prime}(x)$ and so $r(x)=\Bar{h}(x).g(x)$ .  Thus,  $(r(x)\vert 0)=\Bar{h}(x)\cdot(g(x)\vert 0)$.  This implies that $Ker(\sigma)$ is a left $\mathrm{R}[x;\theta_i]$-submodule of $C$ generated by $(g(x)\vert 0)$. Now, by the first isomorphism theorem,
\begin{align*}
\frac{C}{Ker(\sigma)}\cong Im(\sigma)=\langle h(x)\rangle .
\end{align*}
Further,  there exists an element  $(l(x)\vert h(x))\in C$  such that $\sigma(\langle (l(x)\vert h(x))\rangle )=\langle h(x)\rangle $.  This shows that any left $\mathrm{R}[x;\theta_i]$-submodule of $\mathcal{R}_{r,  s}$ is generated by two elements $(g(x)\vert 0)$ and $(l(x)\vert h(x))$.  Hence, we can write
\begin{align*}
C=\langle (g(x)\vert 0),  (l(x)\vert h(x))\rangle .
\end{align*}
\end{proof}

\begin{prop}\label{D4}
Let $C=\langle (g(x)\vert 0),  (l(x)\vert h(x))\rangle = \langle (g_v(x)v+g_{v^\prime}(x)v^\prime \vert 0), ( l_v(x)v+l_{v^\prime}(x)v^\prime \vert  h_v(x)v+h_{v^\prime}(x)v^\prime)  \rangle$ be a double skew cyclic code of length $(r,  s)$ over $\mathrm{R}$.  Then, $deg(l_v(x))<deg(g_v(x))$ and $deg(l_{v^\prime}(x) <deg(g_{v^\prime}(x))$.
\end{prop}

\begin{proof} Without loss of generality,
 assume that $deg(l_v(x))\geq deg(g_v(x))$ and $deg(l_v(x))-deg(g_v(x))=m$,  $m \geq 0$.  Consider the code $D$ generated by
\begin{align*}
\langle (g_v(x)v+g_{v^\prime}(x)v\vert 0),  ((l_v(x)v+l_{v^\prime}(x)v^\prime)-x^m(g_v(x)v+g_{v^\prime}(x)v^\prime)\vert h_v(x)v+h_{v^\prime}(x)v^\prime)\rangle ,
\end{align*}
we have $D\subseteq C$. Further,
\begin{align*}
(l_v(x)v+l_{v^\prime}(x)v^\prime\vert h_v(x)v+h_{v^\prime}(x)v^\prime)&=((l_v(x)v+l_{v^\prime}(x)v^\prime)-x^m(g_v(x)v+g_{v^\prime}(x)v^\prime)\vert h_v(x)v\\ &+h_{v^\prime}(x)v^\prime)+x^m(g_v(x)v+g_{v^\prime}(x)v^\prime\vert 0).
\end{align*}
Therefore,  $C \subseteq D$. Therefore, $C=D$. Hence, $deg(l(x))<deg(g(x))$.
\end{proof}

\begin{prop}\label{D5}
Let $C=\langle (g(x)\vert 0),  (l(x)\vert h(x))\rangle = \langle (g_v(x)v+g_{v^\prime}(x)v^\prime \vert 0), ( l_v(x)v+l_{v^\prime}(x)v^\prime \vert  h_v(x)v+h_{v^\prime}(x)v^\prime)  \rangle$ be a double skew cyclic code of length $(r,  s)$ over $\mathrm{R}$. Then, $g_v(x)\vert_r \frac{x^s-1}{h_v(x)}l_v(x)$ and $g_{v^\prime}(x)\vert_r \frac{x^s-1}{h_{v^\prime}(x)}l_{v^\prime}(x)$. 
\end{prop}

\begin{proof} Since $\frac{x^s-1}{h(x)}\cdot(l(x)\vert h(x))=(\frac{x^s-1}{h(x)}l(x) \vert 0)\in Ker (\sigma)=\langle (g(x)\vert0)\rangle $,  the proof holds.
\end{proof}

\begin{coll}\label{D6}
Let $C=\langle (g(x)\vert 0),  (l(x)\vert h(x))\rangle = \langle (g_v(x)v+g_{v^\prime}(x)v^\prime \vert 0), ( l_v(x)v+l_{v^\prime}(x)v^\prime \vert  h_v(x)v+h_{v^\prime}(x)v^\prime)  \rangle$ be a double skew cyclic code of length $(r,  s)$ over $\mathrm{R}$.  Then, $g_v(x)\vert_r \frac{x^s-1}{h_v(x)}gcd(g_v(x),  l_v(x))$ and $g_{v^\prime}(x)\vert_r \frac{x^s-1}{h_{v^\prime}(x)}gcd(g_{v^\prime}(x),  l_{v^\prime}(x))$.
\end{coll}

\begin{proof} By Theorem \ref{D3},  $h_v(x)\vert_r(x^s-1)$  and $h_{v^\prime}(x)\vert_r(x^s-1)$. This implies that $g_v(x)\vert_r \frac{x^s-1}{h_v(x)}g_v(x)$ and $g_{v^\prime}(x)\vert_r \frac{x^s-1}{h_{v^\prime}(x)}g_{v^\prime}(x)$. Also,  considering Proposition \ref{D5},  we have $g_v(x)\vert_r \frac{x^s-1}{h_v(x)}gcd(g_v(x),  l_v(x))$ and\\ $g_{v^\prime}(x)\vert_r \frac{x^s-1}{h_{v^\prime}(x)}gcd(g_{v^\prime}(x),  l_{v^\prime}(x))$.
\end{proof}

\begin{coll}\label{D7}
Let $C=\langle (g(x)\vert 0),  (l(x)\vert h(x))\rangle = \langle (g_v(x)v+g_{v^\prime}(x)v^\prime \vert 0), ( l_v(x)v+l_{v^\prime}(x)v^\prime \vert  h_v(x)v+h_{v^\prime}(x)v^\prime)  \rangle$ be a double skew cyclic code of length $(r,  s)$ over $\mathrm{R}$. Then, $ lcm(g_v(x),  l_v(x))\vert_r \frac{x^s-1}{h_v(x)}l_v(x)$ and $ lcm(g_{v^\prime}(x),  l_{v^\prime}(x))\vert_r \frac{x^s-1}{h_{v^\prime}(x)}l_{v^\prime}(x)$.
\end{coll}
\begin{proof} Clear.
\end{proof}

\begin{prop}\label{D8}
Let $C$ be a double skew cyclic code of length $(r,  s)$ over $\mathrm{R}$ with $C=\langle (g(x)\vert 0),  (l(x)\vert h(x))\rangle = \langle (g_v(x)v+g_{v^\prime}(x)v^\prime \vert 0), ( l_v(x)v+l_{v^\prime}(x)v^\prime \vert  h_v(x)v+h_{v^\prime}(x)v^\prime)  \rangle$.  Define the sets
\begin{align*}
S_1^v=\bigcup_{i=0}^{r-deg(g_v(x))-1}\lbrace  x^i\cdot(g_v(x)v\vert 0)  \rbrace,
\end{align*}
\begin{align*}
S_1^{v^\prime}=\bigcup_{i=0}^{r-deg(g_{v^\prime}(x))-1}\lbrace  x^i\cdot(g_{v^\prime}(x)v^\prime\vert 0)  \rbrace,
\end{align*}
\begin{align*}
S_2^v=\bigcup_{i=0}^{s-deg(h_v(x))-1}\lbrace  x^i\cdot(l_v(x)v \vert h_v(x)v   \rbrace,
\end{align*}
\begin{align*}
S_2^{v^\prime}=\bigcup_{i=0}^{s-deg(h_{v^\prime}(x))-1}\lbrace  x^i\cdot( l_{v^\prime}(x)v^\prime\vert  h_{v^\prime}(x)v^\prime)  \rbrace.
\end{align*}
Then $S_1^v \cup S_1^{v^\prime}\cup S_2^v \cup S_2^{v^\prime}$ forms a minimal generating set for $C$ as a left $\mathrm{R}[x;\theta_i]$-module.
\end{prop}
\begin{proof} The proof is a slight modification of the proof of Proposition 3.2 in \cite{A5}.
\end{proof}

 Let $C$ be a double skew cyclic code over $\mathrm{R}$ of length $(r,  s)$. We use $C_r$ and $C_s$ to denote the punctured codes of $C$ by deleting the coordinates outside the $r$th and $s$th components, respectively.
 \begin{theo}  \label{D9}
Let $C=\langle (g(x)\vert 0),  (l(x)\vert h(x))\rangle = \langle (g_v(x)v+g_{v^\prime}(x)v^\prime \vert 0), ( l_v(x)v+l_{v^\prime}(x)v^\prime \vert  h_v(x)v+h_{v^\prime}(x)v^\prime)  \rangle $ be a double skew cyclic code of length $(r,  s)$ over $\mathrm{R}$.   Then $C$ is permutation equivalent to an $\mathbb{F}_q$-double skew cyclic code with generating matrix

	$$G= \left( \begin{array}{rrr|rrr}
	I_{r-deg(g_v(x))} & A_1v & A_2v & 0 &0 &0 \\
	0 & B_kv & B_2v & D_1v & I_kv & 0\\
	0 & 0 & 0 & E_1v & E_2v & I_{s-deg(h_v(x))-k_vv}\\
 I_{r-deg(g_{v^\prime}(x))} & A_3{v^\prime} & A_4{v^\prime} & 0 &0 &0 \\
	0 & B_k{v^\prime} & B_4{v^\prime} & D_2{v^\prime} & I_k{v^\prime} & 0\\
	0 & 0 & 0 & E_3{v^\prime} & E_4{v^\prime} & I_{s-deg(h_{v^\prime}(x))-k_{v^\prime}{v^\prime}}
	 \end{array} \right),$$
where $k_v=deg(g_v(x))-deg(gcd_r(g_v(x),  l_v(x)))$ and $k_{v^\prime}=deg(g_{v^\prime}(x))-deg(gcd_r(g_{v^\prime}(x),  l_{v^\prime}(x)))$.
	
 \end{theo}
 \begin{proof} By Proposition \ref{D8},  $C$ is generated by the matrix whose rows are the elements of the set $S_1^v \cup S_1^{v^\prime}\cup S_2^v \cup S_2^{v^\prime}$. Therefore,  $r-deg(g_v(x))$  and $s-deg(h_v(x))$ are the dimension of the matrices generated by the $\theta_i$-shifts of $(g_v(x)\vert 0)$ and  $(l_v(x)\vert h_v(x))$.   Then,  the generating matrix of the code $C$ is permutation equivalent to the following matrix
 	$$G= \left( \begin{array}{rr|rr}
	I_{r-deg(g_v(x))}v & Av & 0 &0  \\
	0 & Bv  & Dv & I_{s-deg(h_v(x))}v\\
 I_{r-deg(g_{v^\prime}(x))}v^\prime & Av^\prime & 0 &0  \\
	0 & Bv^\prime  & Dv^\prime & I_{s-deg(h_{v^\prime}(x))}v^\prime
	 \end{array} \right),$$

 Clearly,  $(C_r)_v$ is  a skew cyclic code generated by $gcd_r(g_v(x),  l_v(x))$. Also,  the submatrix $B$ has rank $deg(g_v(x))-gcd_r(g_v(x),  l_v(x))$.  Moreover,  the generating matrix of $C_r$ is permutation equivalent to the following matrix
 $$
	\begin{pmatrix}
	I_{r-deg(g_v(x))}v  & A_1v & A_2v \\
	0 & B_kv & B_1v\\
	0 & 0 & 0 \\
 I_{r-deg(g_{v^\prime}(x))}v^\prime  & A_1v^\prime & A_2v^\prime \\
	0 & B_kv^\prime & B_1v^\prime\\
	0 & 0 & 0 \\
	\end{pmatrix}
	\quad,
	$$
where $B_k$ is a square matrix of full rank with size $k_v\times k_v$ and $ k_{v^\prime}  \times k_{v^\prime} $, respectively. Hence, we can obtain the above standard form of generating matrix by applying the convenient operations.
  \end{proof}

\section{Duals of $\mathrm{R}$-double skew cyclic codes}
In this section, we present some results on the dual codes of the double skew cyclic codes over $\mathrm{R}$ as a generalization of skew cyclic codes over $\mathrm{R}$ and generator polynomials for the dual of an $\mathrm{R}$-double skew cyclic codes of length $(r,  s)$.

\begin{theo}
Let $C$ be a double skew cyclic code of length $(r,  s)$ over $\mathrm{R}$ with
\begin{align*}
C=\langle (g(x)\vert 0),  (l(x)\vert h(x))\rangle  = \langle (g_v(x)v+g_{v^\prime}(x)v^\prime \vert 0), ( l_v(x)v+l_{v^\prime}(x)v^\prime \vert  h_v(x)v+h_{v^\prime}(x)v^\prime)  \rangle .
\end{align*}
Then $C^\bot\in C_{r,  s}(\mathrm{R})$ and the set
\begin{align*}
C^\bot=\langle (\overline{g}(x)\vert 0),  (\overline{l}(x)\vert \overline{h}(x))\rangle  = \langle (\overline{g}_v(x)v+\overline{g}_{v^\prime}(x)v^\prime \vert 0), ( \overline{l}_v(x)v+\overline{l}_{v^\prime}(x)v^\prime \vert  \overline{h}_v(x)v+\overline{h}_{v^\prime}(x)v^\prime)  \rangle ,
\end{align*}
where $\overline{g}(x),   \overline{l}(x) \in \mathcal{R}_r$,  $\overline{h}(x)\in \mathcal{R}_s$ and $deg(\overline{l}(x))<deg(\overline{g}(x))$.
\end{theo}
\begin{proof} Let $$c=(c_0,  c_1,  \ldots,  c_{r-1}\vert c^\prime_0,  c^\prime_1,  \ldots,  c^\prime_{s-1})\in C;$$ and $$d=(d_0,  d_1,  \ldots,  d_{r-1}\vert d^\prime_0,  d^\prime_1,  \ldots,  d^\prime_{s-1})\in C^\bot.$$ In order to prove $T_{\theta_i}(d)\in C^\bot$, it is enough to show that $<T_{\theta_i}(d),  c>=0$. Let $\gamma=lcm(r,  s)$. Since $C$ is an $\mathrm{R}$-double skew cyclic code, one gets  $T_{\theta_i}^{\gamma-1}(c)=T_{\theta_i}^{\gamma-2}(c)\in C$.

Thus, $<d,  T_{\theta_i}^{\gamma-1}(c)>=0$  for $T_{\theta_i}^{\gamma-1}(c)=(\theta_i(c_{1}),    \ldots,  \theta_i(c_{r-1}),  \theta_i(c_0)\vert \theta_i(c^{\prime}_1),     \ldots,  \theta_i(c^{\prime}_{s-1}), \theta_i(c^{\prime}_0) ).$ According to the definition of the Euclidean inner product and then taken modulo $r$ and $s$,  we have

\begin{align*}
d_0\theta_i(c_1)+\ldots +d_{r-1}\theta_i(c_0)=0;\\
d_0^\prime\theta_i(c^\prime_1)+\ldots + d^\prime_{s-1}\theta_i(c^\prime_0)=0.
\end{align*}

 Applying $\theta_i$ to the above equalities, we have

 \begin{align*}
 \theta_i(d_0)c_1+\ldots + \theta_i(d_{r-2})c_{r-1}+ \theta_i(d_{r-1})c_0=0;\\
  \theta_i(d^\prime_0)c^\prime_1+\ldots +\theta_i(d^\prime_{s-2})c^\prime_{s-1} +\theta_i(d^\prime_{s-1})c^\prime_0=0.
 \end{align*}
 This yields that $\langle T_{\theta_i}(d),  c \rangle=0$, for $T_{\theta_i}(d)=(\theta_i(d_{r-1}),  \theta_i(d_{0}),  \ldots,  \theta_i(d_{r-2}) \vert \theta_i(d^\prime_{s-1}),  \theta_i(d^\prime_{0}),  \ldots,  \theta_i(d^\prime_{s-2}) )$, and hence the proof.

\end{proof}

Set $\xi_r(x):=\sum_{i=0}^{r-1}x^i$ and fix the unit element under the automorphism $\theta_i$.  Then, one can easily obtain the following:
\begin{lem}
 $x^{rs}-1=(x^r-1)\xi_s(x^r)=\xi_s(x^r)(x^r-1)$, where $r,  s\in \mathbb{N}$.
\end{lem}

\begin{defn}
Let $\alpha(x)=(\alpha_{1v}(x)v+\alpha_{1v^\prime}(x)v^\prime \vert \alpha_{2v}(x)v+\alpha_{2v^\prime}(x)v^\prime)$ and $\beta(x)=(\beta_{1v}(x)v+\beta_{1v^\prime}(x)v^\prime \vert
\beta_{2v}(x)v + \beta_{2v^\prime}(x)v^\prime)$ be two elements in $\mathcal{R}_{r,  s}$ and $\gamma=lcm(r,  s)$. Define the map
\begin{align*}
\circ \colon \mathcal{R}_{r,  s} \times \mathcal{R}_{r,  s} &\to \mathcal{R}_{\gamma} \text{ by}\\
\alpha(x)\circ  \beta(x) &  =  (\alpha_{1v}(x)  \vartheta^{\gamma-deg(\beta_{1v}(x))}(\beta_{1v}^\star(x))x^{\gamma-1-deg(\beta_{1v}(x))}\xi_{\frac{\gamma}{r}}(x^r)\\
&+\alpha_{2v}(x)\vartheta^{\gamma-deg(\beta_{2v}(x))}(\beta_{2v}^\star(x))x^{\gamma-1-deg(\beta_{2v}(x))}\xi_{\frac{\gamma}{s}}(x^s))\\
&+ (\alpha_{1v^\prime}(x)  \vartheta^{\gamma-deg(\beta_{1v^\prime}(x))}(\beta_{1v^\prime}^\star(x))x^{\gamma-1-deg(\beta_{1v^\prime}(x))}\xi_{\frac{\gamma}{r}}(x^r)\\
&+\alpha_{2v^\prime}(x)\vartheta^{\gamma-deg(\beta_{2v^\prime}(x))}(\beta_{2v^\prime}^\star(x))x^{\gamma-1-deg(\beta_{2v^\prime}(x))}\xi_{\frac{\gamma}{s}}(x^s)).
\end{align*}
\end{defn}
Clearly, the map $\circ$ is a bilinear map. For the convenience,  we denote $\circ(\alpha(x),  \beta(x))$  by $\alpha(x)\circ \beta(x)$.

\begin{prop}\label{D10}
Let $\alpha=(\alpha_{1,0},  \cdots,  \alpha_{1,r-1} \vert \alpha_{2,0}, \cdots, \alpha_{2,s-1})$ and $\beta=(\beta_{1,0},  \cdots,  \beta_{1,r-1} \vert \beta_{2,0}, \cdots, \beta_{2,s-1})$ be vectors in $\mathrm{R}^r\times \mathrm{R}^s$ with associated polynomials $\alpha(x)=(\alpha_1(x)\vert \alpha_2(x))$ and $\beta(x)=(\beta_1(x)\vert \beta_2(x))$, respectively. Then $\alpha$ and all its $\theta_i$-shifts are orthogonal to $\beta$ if and only if $\alpha(x)\circ \beta(x)=0$.
\end{prop}

\begin{proof} Let $\alpha^{(k)}=(\theta_i^k(\alpha_{1,0-k}), \theta_i^k(\alpha_{1,1-k}), \cdots,  \theta_i^k(\alpha_{1, r-1-k})\vert \theta_i^k(\alpha_{2,0-k}),  \theta_i^k(\alpha_{2,1-k}), \cdots,  \theta_i^k(\alpha_{2, s-1-k}))$ be the $k$-th skew cyclic shift of $\alpha$  where $0\leq k\leq \gamma-1$.  By the definition of orthogonal vectors, we have  $\alpha^{(k)} \cdot  \beta=0$  if and only if $$\sum_{i=0}^{r-1}\theta_i^k(\alpha_{1,  i-k})\beta_{1,i}+\sum_{j=0}^{s-1}\theta_i^k(\alpha_{2, j-k})\beta_{2, j}=0.$$
Let $\Delta_i=\sum_{i=0}^{r-1}\theta_i^k(\alpha_{1, i-k})\beta_{1, i}+\sum_{j=0}^{s-1}\theta_i^k(\alpha_{2, j-k})\beta_{2, j}\in \mathrm{R}$. One can check that
\begin{align*}
\alpha(x)\circ \beta(x)& = [\sum_{\epsilon=0}^{r-1}\sum_{k=\epsilon}^{r-1}\alpha_{1, k-\epsilon}\theta_i^{\gamma-\epsilon}(\beta_{1, k})x^{\gamma-1-\epsilon}+\sum_{\epsilon=1}^{r-1}\sum_{k=\epsilon}^{r-1}\alpha_{1, k}\theta_i^\epsilon(\beta_{1, k-\epsilon})x^{\gamma-1+\epsilon}]\xi_{\frac{\gamma}{r}}(x^r) \\
&+[\sum_{\epsilon=0}^{s-1}\sum_{j=\epsilon}^{s-1}\alpha_{2, j-\epsilon}\theta_i^{\gamma-\epsilon}(\beta_{2,  j})x^{\gamma-1-\epsilon}+\sum_{\epsilon=1}^{s-1}\sum_{j=\epsilon}^{s-1}\alpha_{2, \epsilon}\theta_i^\epsilon(\beta_{2, j-\epsilon}x^{\gamma-1+\epsilon})]\xi_{\frac{\gamma}{s}}(x^{s})\\
&=\sum_{i=0}^{\gamma-1}\theta_i^{\gamma-i}(\Delta_k)x^{\gamma-1-i}.
\end{align*}
Thus, $\alpha(x)\circ \beta(x)=0$ if and only if $\Delta_i=0$ for all $0\leq i \leq \gamma-1$.
\end{proof}

\begin{prop}\label{D11}
Let $\alpha(x)=(\alpha_1(x)\vert \alpha_2(x))$ and $\beta(x)=(\beta_1(x)\vert \beta_2(x))$ be two elements in $\mathcal{R}_{r,  s} $ such that $\alpha(x)\circ \beta(x)=0$. Then
\begin{itemize}
\item[(i)] if $\alpha_1(x)=0$ or $\beta_1(x)=0$, then $\alpha_2(x)\vartheta^{\gamma-deg(\beta_2(x))}(\beta_2^\star(x))\equiv 0 (mod(x^s-1))$;\\
\item[(ii)]if $\alpha_2(x)=0$ or $\beta_2(x)=0$, then $\alpha_1(x)\vartheta^{\gamma-deg(\beta_1(x))}(\beta_1^\star(x))\equiv 0 (mod(x^r-1))$.
\end{itemize}
\end{prop}

 \begin{proof}(ii) Assume that $\alpha_2(x)=0$ or $\beta_2(x)=0$.  Then
 \begin{align*}
 \alpha(x)\circ \beta(x)&=(\alpha_{1v}(x)\vartheta^{\gamma-deg(\beta_{1v}(x))}(\beta_{1v}^\star(x))x^{\gamma-1-deg(\beta_{1v}(x))}\xi_{\frac{\gamma}{r}}(x^r))v\\
 &+(\alpha_{1v^\prime}(x)\vartheta^{\gamma-deg(\beta_{1v^\prime}(x))}(\beta_{1v^\prime}^\star(x))x^{\gamma-1-deg(\beta_{1v^\prime}(x))}\xi_{\frac{\gamma}{r}}(x^r))v^\prime=0(mod(x^\gamma-1)).
 \end{align*}
 Thus, there exists a polynomial $f(x)=f_v(x)v+f_{v^\prime}(x)v^\prime \in \mathrm{R}[x; \theta_i]$ such that
 \begin{align*}
 &(\alpha_{1v}(x)\vartheta^{\gamma-deg(\beta_{1v}(x))}(\beta_1^\star(x))x^{\gamma-1-deg(\beta_{1v}(x))}\xi_{\frac{\gamma}{r}}(x^r))v\\
 & +(\alpha_{1v^\prime}(x)\vartheta^{\gamma-deg(\beta_{1v^\prime}(x))}(\beta_1^\star(x))x^{\gamma-1-deg(\beta_{1v^\prime}(x))}\xi_{\frac{\gamma}{r}}(x^r))v^\prime \\
 &=f_v(x)(x^\gamma-1)v+f_{v^\prime}(x)(x^\gamma-1)v^\prime.
 \end{align*}
 Since $\xi_{\frac{\gamma}{r}}=\frac{x^\gamma-1}{x^r-1}$ and $(x^\gamma-1)(x^r-1)=(x^r-1)(x^\gamma-1)$,  one has
 \begin{align*}
 \alpha_{1v}(x)\vartheta^{\gamma-deg(\beta_{1v}(x))}(\beta_{1v}^\star(x))x^\gamma=f_v(x)x^{deg(\beta_{1v}(x))+1}(x^r-1),
 \end{align*} and
  \begin{align*}
 \alpha_{1v^\prime}(x)\vartheta^{\gamma-deg(\beta_{1v^\prime}(x))}(\beta_{1v^\prime}^\star(x))x^\gamma=f_{v^\prime}(x)x^{deg(\beta_{1v^\prime}(x))+1}(x^r-1).
 \end{align*}
 Thus,
  \begin{align*}
 \alpha_{1v}(x)\vartheta^{\gamma-deg(\beta_{1v}(x))}(\beta_{1v}^\star(x))\equiv 0 \hspace{1mm} (mod(x^r-1)),
 \end{align*} and
  \begin{align*}
 \alpha_{1v^\prime}(x)\vartheta^{\gamma-deg(\beta_{1v^\prime}(x))}(\beta_{1v^\prime}^\star(x))\equiv 0 \hspace{1mm} (mod(x^r-1)).
 \end{align*}
The proof of $(i)$ is similar.
 \end{proof}

Let $C$ be a double skew cyclic code of length $(r,  s)$ over $\mathrm{R}$ with the generating matrix given in Theorem \ref{D9}.  Then the parity check matrix of $C$ is permutation equivalent to
	$$H= \left( \begin{array}{rrr|rrr}
	A_{1}^tv & I_kv & 0 & 0 &B_k^tv & B_k^tvE_2^tv \\
	A_2^tv & 0 & I_{deg(g_v(x))-kv} & 0 & B_1^tv & B_1^tvE_2^tv\\
	0 & 0 & 0 & _{deg(h_v(x))} & D_1^tv & E_1^tv+D_1^tvE_2^tv\\
 A_{1}^tv^\prime & I_kv^\prime & 0 & 0 &B_k^tv^\prime & B_k^tv^\prime E_2^tv^\prime \\
	A_2^tv^\prime & 0 & I_{deg(g_{v^\prime(x)})-kv^\prime} & 0 & B_1^tv^\prime & B_1^tv^\prime E_2^tv^\prime\\
	0 & 0 & 0 & _{deg(h_{v^\prime}(x))} & D_1^tv^\prime & E_1^tv^\prime+D_1^tv^\prime E_2^tv^\prime\\
	 \end{array} \right).$$

Now,  the  following proposition holds:

 \begin{prop}
 Let $C=\langle (g(x)\vert 0),  (l(x)\vert h(x))\rangle = \langle (g_v(x)v+g_{v^\prime}(x)v^\prime \vert 0), ( l_v(x)v+l_{v^\prime}(x)v^\prime \vert  h_v(x)v+h_{v^\prime}(x)v^\prime)  \rangle$ be a double skew cyclic code of length $(r,  s)$ over $\mathrm{R}$ with dual code $C^\bot=\langle (\overline{g}(x)\vert 0),  (\overline{l}(x)\vert \overline{h}(x))\rangle = \langle (\overline{g}_v(x)v+\overline{g}_{v^\prime}(x)v^\prime \vert 0), ( \overline{l}_v(x)v+\overline{l}_{v^\prime}(x)v^\prime \vert  \overline{h}_v(x)v+\overline{h}_{v^\prime}(x)v^\prime)  \rangle$. Then\\
 \begin{itemize}
 \item[(i)] $|C_r|=q^{2r-deg(g_v(x))-deg(g_{v^\prime}(x))+k_v+k_{v^\prime}}$ and $|C_s|=q^{2s-deg(h_v(x))-deg(h_{v^\prime}(x))}$.\\
 \item[(ii)] $|(C^\bot)_r|=q^{deg(g_v(x))+deg(g_{v^\prime}(x))}$ and $|(C^\bot)_s|=q^{deg(h_v(x))+deg(h_{v^\prime}(x))+k_v+k_{v^\prime}}$.\\
 \item[(iii)] $|(C_r)^\bot|=q^{deg(g_v(x))+deg(g_{v^\prime}(x))-k_v-k_{v^\prime}}$ and $|(C_s)^\bot)|=q^{deg(h_v(x))+deg(h_{v^\prime}(x))}$.\\
 \end{itemize}
 \end{prop}

 \begin{coll}\label{D12}
 Let $C=\langle (g(x)\vert 0),  (l(x)\vert h(x))\rangle = \langle (g_v(x)v+g_{v^\prime}(x)v^\prime \vert 0), ( l_v(x)v+l_{v^\prime}(x)v^\prime \vert  h_v(x)v+h_{v^\prime}(x)v^\prime)  \rangle$ be a double skew cyclic code of length $(r,  s)$ over $\mathrm{R}$ with dual code $C^\bot=\langle (\overline{g}(x)\vert 0),  (\overline{l}(x)\vert \overline{h}(x))\rangle = \langle (\overline{g}_v(x)v+\overline{g}_{v^\prime}(x)v^\prime \vert 0), ( \overline{l}_v(x)v+\overline{l}_{v^\prime}(x)v^\prime \vert  \overline{h}_v(x)v+\overline{h}_{v^\prime}(x)v^\prime)  \rangle$. Then\\
\begin{itemize}
 \item[(i)]$deg(\overline{g}_v(x))=r-deg(gcd(g_v(x), l_v(x)))$.\\
  \item[(ii)]$deg(\overline{g}_{v^\prime}(x))=r-deg(gcd(g_{v^\prime}(x), l_{v^\prime}(x)))$.\\
 \item[(iii)]  $deg(\overline{h}_v(x))=s-deg(h_v(x))-deg(g_v(x))+deg(gcd(g_v(x),  l_v(x))).$\\
\item[(iv)]  $deg(\overline{h}_{v^\prime}(x))=s-deg(h_{v^\prime}(x))-deg(g_{v^\prime}(x))+deg(gcd(g_{v^\prime}(x),  l_{v^\prime}(x))).$
\end{itemize}
 \end{coll}

 \begin{theo}
Let $C=\langle (g(x)\vert 0),  (l(x)\vert h(x))\rangle = \langle (g_v(x)v+g_{v^\prime}(x)v^\prime \vert 0), ( l_v(x)v+l_{v^\prime}(x)v^\prime \vert  h_v(x)v+h_{v^\prime}(x)v^\prime)  \rangle$ be a double skew cyclic code of length $(r,  s)$ over $\mathrm{R}$ with dual code $C^\bot=\langle (\overline{g}(x)\vert 0),  (\overline{l}(x)\vert \overline{h}(x))\rangle = \langle (\overline{g}_v(x)v+\overline{g}_{v^\prime}(x)v^\prime \vert 0), ( \overline{l}_v(x)v+\overline{l}_{v^\prime}(x)v^\prime \vert  \overline{h}_v(x)v+\overline{h}_{v^\prime}(x)v^\prime)  \rangle$. Then $$\overline{g}(x)=\frac{x^r-1}{gcd(\vartheta^{\gamma-deg(g_v(x))}(g_v^\star(x)), \vartheta^{\gamma-deg(l_v(x))}(l_v^\star(x)))}v+\frac{x^r-1}{gcd(\vartheta^{\gamma-deg(g_{v^\prime}(x))}(g_{v^\prime}^\star(x)), \vartheta^{\gamma-deg(l_{v^\prime}(x))}(l_{v^\prime}^\star(x)))}v^\prime.$$
 \end{theo}

 \begin{proof} By the generators of dual code $C^\bot$,  $(\overline{g}(x)\vert0)\in C^\bot$.  By Proposition \ref{D10}, we have $$(\overline{g}(x)\vert0)\circ ( g(x)\vert0) \equiv 0 (mod(x^\gamma-1)),$$  $$(\overline{g}(x)\vert0)\circ ( l(x)\vert h(x)) \equiv 0 (mod(x^\gamma-1)).$$ Hence, by Proposition \ref{D11}, $$\overline{g}(x)\vartheta^{\gamma-deg(g(x)
 )}(g^\star(x))\equiv 0(mod(x^r-1)),$$ $$\overline{g}(x)\vartheta^{\gamma-deg(l(x))}(l^\star(x))\equiv 0(mod(x^r-1)),$$  if and only if $$(x^r-1)\vert \overline{g}(x)\vartheta^{\gamma-deg(g(x)
 )}(g^\star(x)),$$  $$(x^r-1)\vert \overline{g}(x)\vartheta^{\gamma-deg(l(x))}(l^\star(x)),$$ respectively. Thus, we have $$(x^r-1)\vert \overline{g}(x)gcd(\vartheta^{\gamma-deg(g(x) )}(g^\star(x)),  \vartheta^{\gamma-deg(l(x))}(l^\star(x)))$$
 if and only if
 $$(x^r-1)\vert \overline{g}_v(x)gcd(\vartheta^{\gamma-deg(g_v(x) )}(g_v^\star(x)),  \vartheta^{\gamma-deg(l_v(x))}(l_v^\star(x)))\text{ and}$$

 $$(x^r-1)\vert \overline{g}_{v^\prime}(x)gcd(\vartheta^{\gamma-deg(g_{v^\prime}(x) )}(g_{v^\prime}^\star(x)),  \vartheta^{\gamma-deg(l_{v^\prime}(x))}(l_{v^\prime}^\star(x))).$$

 This implies that $$\overline{g}_v(x)gcd(\vartheta^{\gamma-deg(g_v(x) )}(g_v^\star(x)),  \vartheta^{\gamma-deg(l_v(x))}(l_v^\star(x)))\equiv 0 (mod(x^r-1)),$$

 $$\overline{g}_{v^\prime}(x)gcd(\vartheta^{\gamma-deg(g_{v^\prime}(x) )}(g_{v^\prime}^\star(x)),  \vartheta^{\gamma-deg(l_{v^\prime}(x))}(l_{v^\prime}^\star(x)))\equiv 0 (mod(x^r-1)).$$

 Therefore,  there exists $f(x)=f_v(x)v+f_{v^\prime}(x)v^\prime\in \mathrm{R}[x;\theta_i]$ such that

 \begin{align*}
   &  \overline{g}_v(x) gcd(\vartheta^{\gamma-deg(g_v(x) )}  (g_v^\star(x)),  \vartheta^{\gamma-deg(l_v(x))}(l_v^\star(x)))v \\
 &+\overline{g}_{v^\prime}(x)gcd(\vartheta^{\gamma-deg(g_{v^\prime}(x) )}(g^\star(x)),  \vartheta^{\gamma-deg(l_{v^\prime}(x))}(l_{v^\prime}^\star(x)))v^\prime\\
 &=f_v(x)v+f_{v^\prime}(x)v^\prime(mod(x^r-1)).
\end{align*}

 It follows from Corollary \ref{D12} that $f_v(x)=f_{v^\prime}(x)=1$ and proof holds.
 \end{proof}

  \begin{theo}
Let $C=\langle (g(x)\vert 0),  (l(x)\vert h(x))\rangle = \langle (g_v(x)v+g_{v^\prime}(x)v^\prime \vert 0), ( l_v(x)v+l_{v^\prime}(x)v^\prime \vert  h_v(x)v+h_{v^\prime}(x)v^\prime)  \rangle$ be a double skew cyclic code of length $(r,  s)$ over $\mathrm{R}$ with dual code $C^\bot=\langle (\overline{g}(x)\vert 0),  (\overline{l}(x)\vert \overline{h}(x))\rangle = \langle (\overline{g}_v(x)v+\overline{g}_{v^\prime}(x)v^\prime \vert 0), ( \overline{l}_v(x)v+\overline{l}_{v^\prime}(x)v^\prime \vert  \overline{h}_v(x)v+\overline{h}_{v^\prime}(x)v^\prime)  \rangle$. Then  $$\overline{h}(x)=\frac{x^s-1}{\vartheta^{\gamma-deg(j(x))}(j^\star(x))}=\frac{x^s-1}{\vartheta^{\gamma-deg(j_v(x))}(j_v^\star(x))}v+\frac{x^s-1}{\vartheta^{\gamma-deg(j_{v^\prime}(x))}(j_{v^\prime}^\star(x))}v^\prime,$$  where $j_v(x)=\frac{lcm(g_v(x),  l_v(x))h_v(x)}{l_v(x)}$ and $j_{v^\prime}(x)=\frac{lcm(g_{v^\prime}(x),  l_{v^\prime}(x))h_{v^\prime}(x)}{l_{v^\prime}(x)}$.
 \end{theo}
 \begin{proof} Let  $(0|j(x))=k(x)(l(x)|h(x))-k(x)(g(x)|0)$,  where $k(x)=\frac{lcm(g(x),l(x))}{l(x)}$,  then $(0|j(x))\in C$.  By Proposition \ref{D10}, we have
$$(\overline{l}(x)\vert \overline{h}(x))\circ (0\vert j(x)) \equiv 0 (mod(x^\gamma-1)).$$ By Proposition \ref{D11}, there exists $f(x)\in \mathrm{R}[x; \theta_i]$ such that $$\overline{h}(x)\vartheta^{\gamma-deg(j(x))}(j^\star(x))=f(x)(x^s-1).$$ Also, by \cite{A8}  and corollary  \ref{D12},
\begin{align*}
deg(\overline{h}_v(x))&=s-deg(h_v(x))-deg(g_v(x))+deg(gcd(g_v(x), l_v(x)))\\
&=s-deg(h_v(x))+deg(l_v(x))-deg(lcm(g_v(x), l_v(x))),
\end{align*}
and
\begin{align*}
deg(\overline{h}_{v^\prime}(x))&=s-deg(h_{v^\prime}(x))-deg(g_{v^\prime}(x))+deg(gcd(g_{v^\prime}(x), l_{v^\prime}(x)))\\
&=s-deg(h_{v^\prime}(x))+deg(l_{v^\prime}(x))-deg(lcm(g_{v^\prime}(x), l_{v^\prime}(x))).
\end{align*}
Hence,
\begin{align*}
deg(x^s-1)=s&=deg(\overline{h}_v(x)\vartheta^{\gamma-deg(j_v(x))}(j_v^\star(x)))\\
&=deg(f_v(x)(x^s-1)),
\end{align*}
and
\begin{align*}
deg(x^s-1)=s&=deg(\overline{h}_{v^\prime}(x)\vartheta^{\gamma-deg(j_{v^\prime}(x))}(j_{v^\prime}^\star(x)))\\
&=deg(f_{v^\prime}(x)(x^s-1)).
\end{align*}
So,  $f_v(x)=f_{v^\prime}(x)=1$ and proof holds.
 \end{proof}

   \begin{theo}
Let $C=\langle (g(x)\vert 0),  (l(x)\vert h(x))\rangle = \langle (g_v(x)v+g_{v^\prime}(x)v^\prime \vert 0), ( l_v(x)v+l_{v^\prime}(x)v^\prime \vert  h_v(x)v+h_{v^\prime}(x)v^\prime)  \rangle$ be a double skew cyclic code of length $(r,  s)$ over $\mathrm{R}$ with dual code $C^\bot=\langle (\overline{g}(x)\vert 0),  (\overline{l}(x)\vert \overline{h}(x))\rangle = \langle (\overline{g}_v(x)v+\overline{g}_{v^\prime}(x)v^\prime \vert 0), ( \overline{l}_v(x)v+\overline{l}_{v^\prime}(x)v^\prime \vert  \overline{h}_v(x)v+\overline{h}_{v^\prime}(x)v^\prime)  \rangle$.  Then,  $$\overline{l}(x)=(\frac{(x^r-1)}{\vartheta^{\gamma-deg(g_v(x))}(g_v^\star(x))}v+\frac{(x^r-1)}{\vartheta^{\gamma-deg(g_{v^\prime}(x))}(g_{v^\prime}^\star(x))}v^\prime)\eta(x),$$  where $\eta(x)\in \mathrm{R}[x; \theta_i]$.
 \end{theo}

 \begin{proof} Let $\overline{c}(x)\in \langle (\overline{g}(x)\vert 0), (\overline{l}(x)\vert \overline{h}(x)) $,  so $\overline{c}(x)\in C^\bot$.  $$\overline{c}(x)\circ (g(x)\vert 0)=\langle (\overline{g}(x)\vert0)\circ (g(x)\vert 0),  (\overline{l}(x)\vert \overline{h}(x))  \circ (g(x)\vert 0) \rangle .$$  So $$\overline{l}(x)\vartheta^{\gamma-deg(g(x))}(g^\star(x))\equiv 0 (mod(x^r-1)).$$ Thus,  $$\overline{l}(x)\vartheta^{\gamma-deg(g(x))}(g^\star(x))=\eta(x)(x^r-1),$$  for some skew polynomial $\eta(x)\in \mathrm{R}[x; \theta_i]$. 

 \end{proof}

   \section{Computational Results and Optimal Codes}

This section introduces a new construction using the generator matrix of double-skew cyclic codes $C$. As an application of our discussion, we present essential examples of double skew cyclic codes of length $(r,  s)$ over the ring $\mathrm{R}$. All the obtained codes have optimal parameters according to the database \cite{A18}. Also, by considering new construction, we obtained better codes than the existing codes in the literature. All the computations are performed using Magma Computational Algebra System \cite{magma}.\\

  \textbf{Construction}: \label{construction}
 Let $C$ be a double skew cyclic code of length $(r, s)$ over $\mathrm{R}$. Now, we construct a linear code $C^\prime$ of length $2(r+s)$ using the generator matrix $G^\prime$ given by
\begin{enumerate} \label{C1}
    \item if $r=s$, $G'=\left[\begin{array}{cc}
        G& G \\
        L &H
    \end{array}\right]$,
    \item if $r<s$, $G'=\left[\begin{array}{cc}
        G& G_1 \\
        L &H
    \end{array}\right]$, where $G_1$ is the lengthening of $G$ upto the length $2s$;
    \item if $r>s$, $G'=\left[\begin{array}{cc}
        G& G_2 \\
        L &H
    \end{array}\right]$, where $G_2$ is the first $2s$ columns of matrix $G$,
\end{enumerate}
where $G$ and $[L|H]$ be generator matrices corresponding to the generating polynomials $\langle g(x)\rangle=\langle g_v(x)v+g_{v^\prime}(x)v^\prime \rangle $ and $\langle l(x)=l_v(x)v+l_{v^\prime}(x)v^\prime )  \vert  h(x)=h_v(x)v+h_{v^\prime}(x)v^\prime \rangle$, respectively.


\begin{ex}
Let $\mathbb{F}_{3^3}$ be the field of order $27$ and $t$ be a primitive root of unity in $\mathbb{F}_{3^3}$. Let $\mathcal{R}_r=R[x;\theta_1]/\langle x^6-1\rangle$ and $\mathcal{R}_s=R[x;\theta_1]/\langle x^3-1\rangle$, where $R=\mathbb{F}_{27}[v]/\langle v^2-v \rangle$ and $\theta_1$ is the Frobenius automorphism extended to the ring $R$. Consider the following factorization in $\mathbb{F}_{3^3}[x;\theta_1]$:
\begin{align*}
    x^6-1&=( x^3 + t^4x^2 + x + t^{14})( x^3 + t^{17}x^2 + t^{22}x + t^{25}),\\
    x^6-1&=( x^3 + t^6x^2 + t^{21}x + 2)( x^3 + t^{19}x^2 + t^{21}x + 1),\\
    x^3-1&=(x^5 + t^4x^4 + t^{14}x^3 + x^2 + t^4x + t^{14})(x + t^{25}),\\
    x^3-1&=(x^5 + t^2x^4 + t^{20}x^3 + x^2 + t^2x + t^{20})( x + t^{19}).
\end{align*}
Let  $C=\langle (g_v(x)v+g_{v^\prime}(x)v^\prime \vert 0), ( l_v(x)v+l_{v^\prime}(x)v^\prime \vert  h_v(x)v+h_{v^\prime}(x)v^\prime)  \rangle $ be the double skew cyclic code of length $(6,  3)$ over $\mathrm{R}$ where
\begin{align*}
   & g_v(x)=x^3 + t^{17}x^2 + t^{22}x + t^{25},~ g_{v^\prime}(x) =  x^3 + t^{19}x^2 + t^{21}x + 1\\
    & l_v(x)=x^2 + t^2x + t,~ l_{v^\prime}(x)=x^2 + t^5x + t^2\\
    & h_v(x)=x + t^{25},  h_{v^\prime}(x)= x + t^{19}.
\end{align*}

 Then, using the Gray map, the standard form of the generator matrix of $C$ is

$$G=\left[ \begin{array}{cccccccccccccccccc}
1& 0& 0& 0& 0& 0& 0& 0& 0& 0& t^{10}& t^{22}& t^6& t^{23}& t^3& t^{19}& 0& t^{22}\\
0& 1& 0& 0& 0& 0& 0& 0& 0& 0& t^{11}& t^6& t^{21}& t^7& t^8& t^2& 0& t^8\\
0& 0& 1& 0& 0& 0& 0& 0& 0& 0& 1& t^6& t^{22}& t^{14}& t^2& 0& t^5& t^9\\
0& 0& 0& 1& 0& 0& 0& 0& 0& 0& 2& 1& t^3& t^{15}& 0& t^{23}& t^{11}& t^3\\
0& 0& 0& 0& 1& 0& 0& 0& 0& 0& t^{21}& t^{16}& t^2& t^{10}& t^3& t^{19}& t^8& t^{21}\\
0& 0& 0& 0& 0& 1& 0& 0& 0& 0& t^{11}& t^8& t^{23}& t^{24}& t& t^{20}& t^5& 1\\
0& 0& 0& 0& 0& 0& 1& 0& 0& 0& 2& t^5& t^{25}& t^{20}& t^2& t^6& t^7& t^6\\
0& 0& 0& 0& 0& 0& 0& 1& 0& 0& t^3& 0& t^7& t^{11}& t^4& t& t^{18}& t\\
0& 0& 0& 0& 0& 0& 0& 0& 1& 0& t^{22}& t^{22}& t^{24}& t^3& t^{17}& t^{12}& t^3& t^{24}\\
0& 0& 0& 0& 0& 0& 0& 0& 0& 1& t^{22}& t^7& t^2& 1& t^9& t^{14}& t^{23}& t^2\\
\end{array}\right].$$
Thus, $C$ has optimal parameters $[18, 10, 6]_{27}$. Now, by the construction \ref{construction}, we present a generator matrix of $C^\prime$ as follows
$$G'=\left[ \begin{array}{cccccccccccccccccc}
1& 0& 0& 0& 0& 0& 0& 0& 0& 0& t^{10}& t^{22}& t^5& t^{22}& t^2& t^{17}& 0& t^{22}\\
0& 1& 0& 0& 0& 0& 0& 0& 0& 0& t^{11}& t^6& t^7& t^{17}& 2& t^{25}& 0& t^8\\
0& 0& 1& 0& 0& 0& 0& 0& 0& 0& 1& t^6& t^7& t^2& t^{18}& t^{15}& t^5& t^9\\
0& 0& 0& 1& 0& 0& 0& 0& 0& 0& 2& 1& t^{10}& 1& t^{24}& t^{24}& t^{11}& t^3\\
0& 0& 0& 0& 1& 0& 0& 0& 0& 0& t^{21}& t^{16}& 0& t^7& t^{20}& t^{17}& t^{15}& t^{21}\\
0& 0& 0& 0& 0& 1& 0& 0& 0& 0& t^{11}& t^8& t^5& t^7& t^2& t^{25}& t^5& 2\\
0& 0& 0& 0& 0& 0& 1& 0& 0& 0& 2& t^5& t^{11}& 0& t^{11}& t^9& t^7& t^6\\
0& 0& 0& 0& 0& 0& 0& 1& 0& 0& t^3& 0& t^{14}& t^{12}& t& t^{23}& t^{18}& t\\
0& 0& 0& 0& 0& 0& 0& 0& 1& 0& t^{22}& t^{22}& t^{15}& t^{12}& t^{19}& t^{19}& t^3& t^{24}\\
0& 0& 0& 0& 0& 0& 0& 0& 0& 1& t^{22}& t^7& t^2& t^{18}& t^{17}& t^3& t^{23}& t^2\\
\end{array}\right].$$
Hence, $C^\prime$ has optimal parameters $[18, 10, 7]_{27}$ which is better than the code $[18, 10, 6]_{27}$ given in \cite{A7}.

 \end{ex}

\begin{ex}
Let $\mathbb{F}_{2^4}$ be the field of order $16$ where $t$ is a primitive root of unity in $\mathbb{F}_{2^4}$. Let $R_r=R_s=R[x;\theta_1]/\langle x^8-1\rangle$ where $R=\mathbb{F}_{16}[v]/\langle v^2-v \rangle$ and $\theta_1$ is the Frobenius automorphism extended to the ring $R$. Consider the following factorization in $\mathbb{F}_{16}[x;\theta_1]$
\begin{align*}
    x^8-1&=(x^4 + t^3x^3 + t^7x^2 + t^4x + 1)(x^4 + t^3x^3 + t^2x^2 + t^4x + 1)\\
    x^8-1&=(x^6 + t^9x^5 + t^7x^4 + t^2x^3 + tx^2 + t^6x + t^5)( x^2 + t^6x + t^{10})\\
    x^8-1&=(x^6 + t^8x^5 + t^4x^4 + tx^3 + t^{14}x^2 + t^5x + t^6)( x^2 + t^2x + t^9).
\end{align*}
Let  $C=\langle (vg_v(x)+(1-v)g_{v^\prime}(x)|0),(vl_v(x)+(1-v)l_{v^\prime}(x)|vh_v(x)+(1-v)h_{v^\prime}(x))  \rangle$ be the double skew cyclic code of length $(8,  8)$ where
\begin{align*}
    &g_v(x)=g_{v^\prime}(x) = x^4 + t^3x^3 + t^2x^2 + t^4x + 1,\\
    &l_v(x)=l_{v^\prime}(x)=x^3 + t^7x^2 + t^3x + t,\\
    &h_v(x)=x^2 + t^6x + t^{10},\\
    &h_{v^\prime}(x)=x^2 + t^2x + t^9.
\end{align*}
 Then, using the Gray map, the standard form of the generator matrix of $C$ is:

$$G=\left[\begin{array}{cccccccccccccccccccccc}
     1& 0& 0& 0& 0& 0& 0& t^3& 0& 0& 0& 0& 0& 0& 0& 0& 0& 0& 0& 0& 0& 0\\
0& 1& 0& 0& 0& 0& 0& t^9& 0& 0& 0& 0& 0& 0& 0& 0& 0& 0& 0& 0& 0& 0\\
0& 0& 1& 0& 0& 0& 0& t^6& 0& 0& 0& 0& 0& 0& 0& 0& 0& 0& 0& 0& 0& 0\\
0& 0& 0& 1& 0& 0& 0& 1& 0& 0& 0& 0& 0& 0& 0& 0& 0& 0& 0& 0& 0& 0\\
0& 0& 0& 0& 1& 0& 0& t^3& 0& 0& 0& 0& 0& 0& 0& 0& 0& 0& 0& 0& 0& 0\\
0& 0& 0& 0& 0& 1& 0& t^9& 0& 0& 0& 0& 0& 0& 0& 0& 0& 0& 0& 0& 0& 0\\
0& 0& 0& 0& 0& 0& 1& t^6& 0& 0& 0& 0& 0& 0& 0& 0& 0& 0& 0& 0& 0& 0\\
0& 0& 0& 0& 0& 0& 0& 0& 1& 0& 0& 0& 0& 0& 0& t^3& 0& 0& 0& 0& 0& 0\\
0& 0& 0& 0& 0& 0& 0& 0& 0& 1& 0& 0& 0& 0& 0& t^9& 0& 0& 0& 0& 0& 0\\
0& 0& 0& 0& 0& 0& 0& 0& 0& 0& 1& 0& 0& 0& 0& t^6& 0& 0& 0& 0& 0& 0\\
0& 0& 0& 0& 0& 0& 0& 0& 0& 0& 0& 1& 0& 0& 0& 1& 0& 0& 0& 0& 0& 0\\
0& 0& 0& 0& 0& 0& 0& 0& 0& 0& 0& 0& 1& 0& 0& t^3& 0& 0& 0& 0& 0& 0\\
0& 0& 0& 0& 0& 0& 0& 0& 0& 0& 0& 0& 0& 1& 0& t^9& 0& 0& 0& 0& 0& 0\\
0& 0& 0& 0& 0& 0& 0& 0& 0& 0& 0& 0& 0& 0& 1& t^6& 0& 0& 0& 0& 0& 0\\
0& 0& 0& 0& 0& 0& 0& 0& 0& 0& 0& 0& 0& 0& 0& 0& 1& 0& 0& 0& 0& 0\\
0& 0& 0& 0& 0& 0& 0& 0& 0& 0& 0& 0& 0& 0& 0& 0& 0& 1& 0& 0& 0& 0\\
0& 0& 0& 0& 0& 0& 0& 0& 0& 0& 0& 0& 0& 0& 0& 0& 0& 0& 1& 0& 0& 0\\
0& 0& 0& 0& 0& 0& 0& 0& 0& 0& 0& 0& 0& 0& 0& 0& 0& 0& 0& 1& 0& 0\\
0& 0& 0& 0& 0& 0& 0& 0& 0& 0& 0& 0& 0& 0& 0& 0& 0& 0& 0& 0& 1& 0\\
0& 0& 0& 0& 0& 0& 0& 0& 0& 0& 0& 0& 0& 0& 0& 0& 0& 0& 0& 0& 0& 1\\
\end{array}\right|
$$
$$
\left|\begin{array}{cccccccccc}
      t^4& t& t^{5}& t^3& t^6& t^{4}& t^{6}& t^5 & t^3& t^{2}\\
 t^{11}& 0& t& t^9& t^{11}& t& t& t^{2} & t^7& t^{12}\\
 t^{6}& 1& t^{12}& t^6& t& t^{8}& t^3& t^8 & t^{12}& t^{14}\\
 t& t& t^4& 0& t^{13}& t^{6}& t^{11}& t^{10}& t^{12} & t^7\\
t^{2}& t^8& t& 0& t^{10}& t^3& t^{8} & 0& t^{11}& t^4\\
 t^2& t^{9}& t^9& t^4& t& t^{13}& t^2& t^{12}& t^7& t^{13}\\
 t^5& t^7& t^{12}& t^{11}& t^{11}& t^{10}& t^6& 1& t & t^2\\
 t& t^5& t^8& t^5& t^{11}& t^{8}& t^{5}& t^9 & 1& t^3\\
 1& t^{7}& t^{10}& t^9& t^{11}& t^7& 1 & t& t^2& t^9\\
t^5& t^4& t^9& t& 0& t^{13}& t^{12}& t^{9} & t^2& t^5\\
 t^{4}& t^4& 1& t^7& t& t^{14}& t^8& t & t& t^{3}\\
 1& t^{12}& t^9& t^{2}& t^6& t^{11}& t& 1 & t^{13}& t\\
 t^8& t^{14}& 1& t^7& t^{10}& t^{7}& t^{4} & t^{6}& t^5& t^{10}\\
 t^4& t^{13}& t^{4}& t^{4}& t^5& t^6& t^{13}& t^{4} & t^{3}& t^{2}\\
 t^{3}& t^9& t^4& t^{11}& t& 0& t^3& 1 & t^8& t^{13}\\
 t^{13}& t^{14}& t^{5}& t^5& 1& t^5& t^{7}& t^{7}& t^{12}& t^7\\
 t^{11}& t^{2}& t^{11}& 1& t& t& t^{8}& t^{11} & t^{4}& t^{4}\\
 t^6& t^8& 1& t^6& t^3& 0& t^4& t^7& t^4 & 1\\
t& t& t^{5}& t^{12}& t^{3}& t^{2}& t^{11}& t^{10} & t^6& t^{7}\\
t^{3}& t^{6}& t& t^3& t^8& 0& t^{10}& t& t^{12} & 1\\
\end{array}\right].$$

Then, the parameters of code $C$ are $[32,20,4]_{16}$. Now, by the construction \ref{construction}, the generator matrix of new code $C'$ is given by

$$G'=\left[\begin{array}{cccccccccccccccccccccc}
     1& 0& 0& 0& 0& 0& 0& t^3& 0& 0& 0& 0& 0& 0& 0& 0& 0& 0& 0& 0& 0& 0\\
0& 1& 0& 0& 0& 0& 0& t^9& 0& 0& 0& 0& 0& 0& 0& 0& 0& 0& 0& 0& 0& 0\\
0& 0& 1& 0& 0& 0& 0& t^6& 0& 0& 0& 0& 0& 0& 0& 0& 0& 0& 0& 0& 0& 0\\
0& 0& 0& 1& 0& 0& 0& 1& 0& 0& 0& 0& 0& 0& 0& 0& 0& 0& 0& 0& 0& 0\\
0& 0& 0& 0& 1& 0& 0& t^3& 0& 0& 0& 0& 0& 0& 0& 0& 0& 0& 0& 0& 0& 0\\
0& 0& 0& 0& 0& 1& 0& t^9& 0& 0& 0& 0& 0& 0& 0& 0& 0& 0& 0& 0& 0& 0\\
0& 0& 0& 0& 0& 0& 1& t^6& 0& 0& 0& 0& 0& 0& 0& 0& 0& 0& 0& 0& 0& 0\\
0& 0& 0& 0& 0& 0& 0& 0& 1& 0& 0& 0& 0& 0& 0& t^3& 0& 0& 0& 0& 0& 0\\
0& 0& 0& 0& 0& 0& 0& 0& 0& 1& 0& 0& 0& 0& 0& t^9& 0& 0& 0& 0& 0& 0\\
0& 0& 0& 0& 0& 0& 0& 0& 0& 0& 1& 0& 0& 0& 0& t^6& 0& 0& 0& 0& 0& 0\\
0& 0& 0& 0& 0& 0& 0& 0& 0& 0& 0& 1& 0& 0& 0& 1& 0& 0& 0& 0& 0& 0\\
0& 0& 0& 0& 0& 0& 0& 0& 0& 0& 0& 0& 1& 0& 0& t^3& 0& 0& 0& 0& 0& 0\\
0& 0& 0& 0& 0& 0& 0& 0& 0& 0& 0& 0& 0& 1& 0& t^9& 0& 0& 0& 0& 0& 0\\
0& 0& 0& 0& 0& 0& 0& 0& 0& 0& 0& 0& 0& 0& 1& t^6& 0& 0& 0& 0& 0& 0\\
0& 0& 0& 0& 0& 0& 0& 0& 0& 0& 0& 0& 0& 0& 0& 0& 1& 0& 0& 0& 0& 0\\
0& 0& 0& 0& 0& 0& 0& 0& 0& 0& 0& 0& 0& 0& 0& 0& 0& 1& 0& 0& 0& 0\\
0& 0& 0& 0& 0& 0& 0& 0& 0& 0& 0& 0& 0& 0& 0& 0& 0& 0& 1& 0& 0& 0\\
0& 0& 0& 0& 0& 0& 0& 0& 0& 0& 0& 0& 0& 0& 0& 0& 0& 0& 0& 1& 0& 0\\
0& 0& 0& 0& 0& 0& 0& 0& 0& 0& 0& 0& 0& 0& 0& 0& 0& 0& 0& 0& 1& 0\\
0& 0& 0& 0& 0& 0& 0& 0& 0& 0& 0& 0& 0& 0& 0& 0& 0& 0& 0& 0& 0& 1\\
\end{array}\right|
$$
$$
\left|\begin{array}{cccccccccc}
      t^5& t^8& t^{14}& t^2& 0& t^{10}& t^{11}& t^9 & t^8& t^{14}\\
 t^{11}& t^{11}& t^{12}& 1& t^9& t^9& t^5& t^{10} & t^5& t^{13}\\
 t^{10}& t^7& t^8& t^8& t^4& t^{12}& t^8& t & t^6& t^{14}\\
 t^8& t^{10}& 1& 0& t^{13}& t^{13}& t^9& 1& t^{10} & t^5\\
t^{13}& t^7& t^{10}& t^{12}& t^{10}& t^8& t^{12} & t^5& t^{10}& t^3\\
 t^2& t^{10}& t^5& 1& t^5& t^6& 1& 1& 1& t\\
 t& t& t^{10}& 0& t& t^{10}& t^3& t^7& t^{12} & t^9\\
 t^{14}& t^4& 0& t^6& t^{11}& t^{14}& t^{11}& t^5 & t^{13}& t^2\\
 t^8& t^{12}& t^{10}& t^5& t^{11}& t^7& t^{10} & t^{11}& t^6& t^6\\
t^3& t^5& 1& t^6& t^{13}& t^3& t^9& t^{10} & t^3& t\\
 t^{10}& t^3& 0& t^9& t^4& t^{14}& t^7& t^{10} & t^8& t^{11}\\
 t^{13}& t^5& t^6& t^{12}& 1& t^{13}& t^{12}& t^8 & t& t\\
 1& t^{10}& t^{13}& t^7& t^{11}& t^{13}& t^{14} & t^{10}& t^5& t^6\\
 t^8& t^4& t^{12}& t^{11}& t^9& 0& t& t^{12} & t^{13}& t^{12}\\
 t^{12}& t^8& t^7& t^{14}& t^5& t^4& t^5& t^7 & t^3& t^{14}\\
 t^{14}& t^3& t^{13}& t^6& t^{10}& t& t^{12}& t^{11}& t^{11}& t\\
 t^2& t^{12}& t^8& t^{11}& t^7& t^{11}& t^{13}& t & t^{11}& t^{14}\\
 t^2& t^2& t^{13}& t^9& 0& t^{10}& t& t^2& t & t^7\\
t^{14}& t^{14}& t^{14}& t^4& 1& 1& t^8& 1& t^{10} & t^3\\
t^9& t^3& t^{10}& t^{12}& t^{12}& t^{14}& t^6& t^8 & t^6& t^{14}\\
\end{array}\right]$$

Hence, $C'$ has the parameters $[32, 20, 8]_{16}$ which is better than the code $[32, 16, 8]_{16}$ given in \cite{A7}.

\end{ex}
\newpage
  \renewcommand{\arraystretch}{1.4}
\begin{table}
\label{comparison}
\caption{Optimal $\mathrm{R}$-double skew cyclic codes}
\begin{tabular}{|c|c|c|}
	\hline
$(r,s,p^m)$  &$g_v(x),g_{v^\prime}(x), l_v(x),l_{v^\prime}(x),h_v(x), h_{v^\prime}(x)$&  $\Phi(\mathcal{C})$\\
\hline

 $(6,6,2^4)$&$x^4 + x^2 + 1, x^4 + t^{10}x^2 + t^5, x^3 + t^8x^2 + tx + t^7, $&$[24, 12, 10]$\\
&$x^3 + t^6x^2 + t^4x + t^{11}, x^2 + x + 1, x^2 + t^5x + 1$&\\
	\hline

 $(6,6,2^4)$&$ x^3 + t^5x^2 + t^{10}x + t^{10}, x^3 + t^5x^2 + t^{10}x + t^{10},$&$[24, 14, 6]$\\
&$  x^2 + tx + t^3, x^2 + tx + t^3, x^2 + t^{10}, x^2 + t^{10}x + 1$&\\
	\hline

  $(6,3,3^3)$&$ x^3 + t^{17}x^2 + t^{22}x + t^{25}, x^3 + t^{19}x^2 + t^{21}x + 1,$&$[18, 10, 7]^*$\\
&$ x^2 + t^2x + t, x^2 + t^5x + t^2, x + t^{25}, x + t^{19} $&\\
	\hline

  $(6,3,3^3)$&$x^2 + t^{21}x + t^{25}, x^2 + t^{22}x + t^{11}, x + t^2,x + t^5, x + t^{23}, x + 2$&$[18, 12, 5]$\\
  \hline

  $(6,6,3^2)$&$x^2 + x + 1, x^2 + t^3x + t^6, x + 1, x + t,$&$[24, 16, 6]$\\
&$x^2 + t^3x + t^2, x^2 + t^6x + 1$&\\
	\hline
   $(6,6,2^4)$&$x^2 + 1, x^2 + t^5, x + t, x + t^2, x^2 + t^5, x^2 + t^5x + 1$&$[24, 16, 6]$\\
	\hline
   $(8,8,2^4)$&$x^4 + t^3x^3 + t^2x^2 + t^4x + 1, x^4 + t^3x^3 + t^2x^2 + t^4x + 1, $&$[32, 20, 8]^*$\\
&$x^3 + t^7x^2 + t^3x + t, x^3 + t^7x^2 + t^3x + t, x^2 + t^6x + t^{10}, x^2 + t^2x + t^9$&\\
	\hline
 $(8,8,2^4)$&$x^3 + t^5x^2 + 1, x^3 + t^5x^2 + 1, x^2 + t^3x + t,$&$[32, 22, 6]$\\
&$ x^2 + t^3x + t, x^2 + 1, x^2 + t^{10}x + t^3 $&\\
	\hline
   $(8,8,2^4)$&$x^2 + 1, x^2 + t^2x + 1, x + 1, x + t, x^2 + t^2x + t^5,  x^2 + t^{13}x + t^2$&$[32, 24, 6]$\\
	\hline

$(12,12,2^2)$&$x^3 + 1, x^3 + x^2 + x + 1, x^2 + tx + t^2, x^2 + tx + 1, $&$[48, 36, 6]$\\
&$x^3 + tx^2 + t^2x + 1, x^3 + t^2x^2 + t^2x + t$&\\
	\hline
 $(12,12,3^2)$&$ x^2 + t^7x + t^5, x^2 + t^6x + t^2, x + t, x + 1,$&$[48, 40, 6]$\\
&$ x^2 + t^5x + t^2, x^2 + tx + t^7$&\\
	\hline
  $(12,12,2^4)$&$ x^3 + t^2x^2 + t^{11}x + 1, x^3 + t^3x^2 + t^9x + t^4, x^2 + tx + t^3
, x^2 + t^7x + t^3,$&$[48, 36, 7]$\\
&$ x^3 + t^3x^2 + t^9x + t^4, x^3 + t^9x^2 + t^{13}x + t^4$&\\
	
	\hline
  $(14,14,2^2)$&$x^3 + x^2 + 1, x^3 + t^2x^2 + tx + t^2, x^2 + x + 1,$&$[56, 44, 6]$\\
&$ x^2 + x + t, x^3 + t^2x^2 + t^2x + t^2, x^3 + t^2x^2 + t^2 $&\\
	\hline

  $(24,24,2^2)$&$ x^2 + 1, x^2 + t^2x + 1, x + 1, x + t, x^2 + t^2, x^2 + tx + 1$&$[96, 88, 4]$\\

	\hline
  $(42,42,3^2)$&$x^2 + x + 1, x^2 + 2x + 1, x + t, x + t^2, x^2 + tx + t^2,$&$[168, 160, 4]$\\
&$ x^2 + t^3x + t^6$&\\
	\hline
\end{tabular}
\end{table}

$^*$: better parameters than $[32, 16, 8]_{16}$ and $[18, 10, 6]_{27}$ respectively given in \cite{A7}.

\newpage
\section{Conclusion}
This paper has studied the algebraic structure of double skew cyclic codes of block length $(r,
 s)$ over a finite non-chain ring $\mathrm{R}$. We determined the generators of $\mathrm{R}$-double skew cyclic codes and their duals. Furthermore, we defined a new generator matrix to obtain new codes with better parameters. To support our results, we showed good examples and created a table of some optimal code over the ring $\mathrm{R}$.
 The present work also directs future works in two ways: These codes may be studied over another mixed alphabet or extend this study over a generalized structure and look for the construction of quantum codes as the applications. The former path seems more straightforward than the latter.

\section{Acknowledgment} In this research, the author Ashutosh Singh is thankful to the Council of Scientific \& Industrial Research (CSIR), Govt. of India under File No. 09/1023(0027)/2019- EMR-1 and the author Tulay Yildirim is thankful to the Scientific and Technological Research Council of Turkey (TUBITAK) ARDEB 1002-A Grant No 123F286 for the financial support.

\end{document}